\documentclass[reprint,secnumarabic,amssymb, nobibnotes, aps, prd, superscriptaddress, showkeys, showpacs, longbibliography]{revtex4-1}

\usepackage{amsmath}
\usepackage{amssymb}
\usepackage{mathscinet}
\usepackage{amsfonts,amsmath}  
\usepackage{latexsym} 
\usepackage[utf8]{inputenc}
\usepackage{verbatim}
\usepackage{pb-diagram} 

\newtheorem{theorem}{Theorem}[section]

\newtheorem{corollary}[theorem]{Corollary}

\newtheorem{proposition}[theorem]{Proposition}

\newtheorem{example}[theorem]{Example}

\newenvironment{proof}{\textbf{Proof:}}{\hfill$\square$}

\newcommand{\bes}{\begin{eqnarray*}}
\newcommand{\ees}{\end{eqnarray*}} 
\newcommand{\bpm}{\begin{pmatrix}}
\newcommand{\epm}{\end{pmatrix}}

\newcommand{\diag}{{\rm diag}}
\newcommand{\bra}[1]{\langle #1 |}    
\newcommand{\ket}[1]{| #1 \rangle}    
\newcommand{\bq}{\ket{q}}
\newcommand{\R}{{\mathbb R}}
\newcommand{\IR}{{\mathbb R}}
\newcommand{\IC}{{\mathbb C}}
\begin{document}

\title[Perturbation bounds for state transfer]
{Bounds on probability of state transfer \\ with respect to readout time and edge weight}

\author{Whitney Gordon}
\affiliation{Department of Mathematics \& Computer Science, Brandon University, Brandon,
 MB, Canada R7A 6A9}
\author{Steve Kirkland}
 \affiliation{Department of Mathematics, University of Manitoba, Winnipeg, MB, Canada  R3T 2N2}
\author{Chi-Kwong Li} 
\affiliation{Department of Mathematics, College of William and Mary,  Williamsburg, VA, USA  23187}
\author{Sarah Plosker}
\affiliation{Department of Mathematics \& Computer Science, Brandon University, Brandon,
 MB, Canada R7A 6A9}
 \affiliation{Department of Mathematics, University of Manitoba, Winnipeg, MB, Canada  R3T 2N2}
 \affiliation{ Winnipeg Institute for Theoretical Physics, Winnipeg, Manitoba}
\email{ploskers@brandonu.ca}
\author{Xiaohong Zhang}
 \affiliation{Department of Mathematics, University of Manitoba, Winnipeg, MB, Canada  R3T 2N2}
\date{January 21, 2016}%

\begin{abstract}
We analyse the sensitivity of a spin chain modelled by an undirected weighted connected graph exhibiting perfect state transfer to small perturbations in readout time and edge weight in order to obtain physically relevant bounds on the probability of state transfer. 
 At the heart of our analysis    is the concept of the numerical range of a matrix;  our analysis of   edge weight errors additionally makes use of the spectral and Frobenius norms. 

\end{abstract}
\maketitle
 
  \section{Introduction}
Transmitting a quantum state from one location to another is a critical task within a quantum computer. This task can be realised through the use of a  spin chain (1D magnet). The seminal work by Bose \cite{Bose} describes how a spin chain can be  used as a quantum data bus for quantum communication within a quantum computer. 
This leads to the notion of perfect state transfer (PST, described in more detail below), a desirable property for quantum communication.  
Although the spin chain considered by Bose minimizes the amount of physical and technological resources required to transfer quantum states, it only exhibits PST for $n\leq 3$, where   $n$ is the number of spins in the chain. A spin chain can  be represented by a graph, so many other types of graphs have been considered for quantum state transfer with $n>3$.

The \textit{fidelity} or \textit{probability of state transfer} is a measure of the closeness between two quantum states and is used to determine the accuracy of state transfer through a quantum data bus between quantum registers and/or processors. Fidelity is a number between 0 and 1; when the fidelity between two quantum states is equal to 1 we have \textit{perfect state transfer} (PST), and when the fidelity can be made arbitrarily close to 1 we have \textit{pretty good state transfer} (PGST). Many families of graphs been shown to exhibit PST \cite{CDEL, CDDEKL, BGS08, Ang09, Ang10,  K-Sev,  B13} or PGST \cite{GKSS12}.

In this work, we take a mathematical approach to perturbations which decrease the probability of state transfer.  While our approach is similar in nature to that found in \cite{Kay06} and in \cite{Steve2015}, it should be noted that other authors take different approaches. In particular, there have been a number of numerical studies investigating the robustness of fidelity with respect to perturbations (e.g.\ \cite{numerical1, numerical2, numerical3}), while the consideration of a disordered XY model leads to discussions of the appearance of Anderson localization (see \cite{Andloc}).  A recent paper \cite{Kay15} concerns the use of error correcting codes as a strategy for dealing with imperfections (in contrast, we do not consider encoding/decoding schemes herein).

 We consider state transfer probability as it applies to a weighted, undirected graph $G$, where vertices are labelled $1, \dots, n$ and the weight of the edge between vertices $j$ and $k$ is denoted $w(j,k)$. For any graph $G$, we consider its $n\times n$ adjacency matrix $A=[a_{jk}]$ defined via
 \begin{eqnarray*}
 a_{jk}= \begin{cases}
    w(j,k)      & \quad \text{if } j \text{ and } k\text{ are adjacent}\\
   0 & \quad\text{otherwise }\\
  \end{cases}
 \end{eqnarray*}
 as well as its Laplacian matrix $L=R-A$, where $R$ is the diagonal matrix of row sums of $A$. 
 
 Depending on the dynamics of our system, the Hamiltonian $H$, representing the total energy of our system, is taken to be either $A$ (in the case of $XX$ dynamics) or $L$ (in the case of Heisenberg ($XXX$) dynamics).  Its spectrum represents the possible measurement outcomes when one measures the total energy of the system. Here we are not making full use of the Hamiltonian in that we are taking a snapshot in time---neither $A$ nor $L$ depends on $t$. We account for time by setting $U(t)=e^{itH}$. The fidelity of transfer from vertex $s$ (sender) to vertex $r$ (receiver) is then  given by $p(t)=|u(t)_{sr}|^2$, where $u(t)_{sr}$ is the $(s,r)$-th entry of $U(t)=e^{itH}$. 
 
 Ideally, the fidelity is 1, representing perfect state transfer (PST) between the sender and receiver. 
  In \cite{Kay06}, the author discusses the very issue of tolerance of a spin chain with respect to timing errors and with respect to edge weight errors (so-called \emph{manufacturing} errors). For timing errors, he derives a simple lower bound based on the squared difference between each eigenvalue and the smallest eigenvalue, noting that a Hamiltonian with minimal eigenvalue spread would optimize the bound for small perturbations in readout time. The bounds that we produce (for both the  adjacency and Laplacian cases) look similar   and in fact extend the lower bound given in \cite{Kay06}.  Moreover, we give an example where our bound is attained for the adjacency matrix case, so it cannot be further improved in that setting.  Sensitivity with respect to perturbations in readout time is discussed in Section \ref{sec:time}.
  
  For manufacturing errors, again in \cite{Kay06} the author finds that distances  between eigenvalues are key, although no bound is given. 
  This sensitivity analysis was continued in \cite{Steve2015} through an analysis of the  derivatives of the fidelity of state transfer with respect to either readout time or a fixed $(j,k)$-th edge weight. Again it was noted that minimizing the spectral spread optimizes the bound on the fidelity of state transfer for small perturbations in time. No explicit bound with respect to perturbations in edge weight were given in  \cite{Steve2015}; the edge weight results were more qualitative in nature. Here, we take several different approaches to give bounds on  the probability of state transfer with respect to edge weight perturbation which involve both the spectral and Frobenius norms. 
 Sensitivity with respect to perturbations in edge weights is discussed in Section \ref{sec:edgeweight}.

 \section{Sensitivity with respect to readout time}\label{sec:time}
 Suppose we have PST between vertices $j$ and $k$ at time $t_0$. How sensitive is $p(t_0)$ to small changes in time? We would like $p(t_0+h)$ to be close to $p(t_0)$  for small $h$. 
 
 Let us fix some notation now. Let $\{\ket{1}, \ket{2}, \dots, \ket{n}\}$ denote  the standard ordered basis for $\IC^n$. Let $H$ be a real symmetric matrix that we decompose as $H=Q^T\Lambda Q$, where $\Lambda=\diag(\lambda_1, \dots, \lambda_n)$ is a (real) diagonal matrix of eigenvalues and $Q^T$ is an  orthogonal matrix of corresponding eigenvectors. If we have
 PST between vertices $j$ and $k$, we can permute the columns of  $Q^T$ so that we have PST between vertices 1 and 2. Thus we will always focus on the $(1,2)$ entry of our matrix for simplicity and ease of notation. Let
 $\ket{q_1} = (q_{11}, \dots, q_{n1})^T$ and 
$\ket{q_2} = (q_{12}, \dots, q_{n2})^T$ be the first two columns of $Q$; these vectors represent the first (and the second, respectively) entries of all the eigenvectors of $H$. 
 We then have $e^{it_0H}=Q^Te^{it_0\Lambda}Q$. We are assuming PST between vertices 1 and 2 at time $t_0$, and so
 \begin{eqnarray*}
 \left|\bra{q_1} e^{it_0\Lambda}\ket{q_2}\right|&=&1,\ees
 giving $e^{i\theta}\ket{q_1}=e^{it_0\Lambda}\ket{q_2}$ for some $\theta \in \R$, 
 where (obviously) $|e^{i\theta}|=1$. In particular, we have
 \bes
  \left|\bra{q_1} \tilde{M}e^{it_0\Lambda}\ket{q_2}\right|&=&\left|\bra{q_1}\tilde{M}e^{i\theta}\ket{q_1}\right|=\left|\bra{q_1}{M}\ket{q_1}\right|
 \ees 
 for any matrix $\tilde{M}$ (where we have set $M=\tilde{M}e^{i\theta}$). This innocuous observation will allow  us to consider the $(1,1)$ entry of $Q^TMQ$ rather than the $(1,2)$ entry of $Q^T\tilde{M}e^{it_0\Lambda}Q$ (see the proof of Theorem \ref{thm:time} below). 
 
   The change from the $(1,2)$ entry to the $(1,1)$ entry will forge the link between the sensitivity of the fidelity with respect to readout time and the notion of the \emph{numerical range} of an $n\times n$ matrix 
 $M$, defined by
 \bes
    W(M) = \left\{\bra{x}M\ket{x} \,:\,\ket{x}\in\mathbb{C}^n,\langle x|x\rangle=1\right\}.
\ees

 We now consider a small perturbation of readout time. The motivation for this is highly practical: even with lab equipment calibrated to an arbitrary amount of precision, if we want to readout at time $t_0$, the readout time in practice will be $t_0+h$ for small $h$ (e.g.\ $h=\pm 0.0001$).

\begin{theorem}\label{thm:time}
Let $H$ be either the adjacency matrix or the Laplacian associated with an undirected weighted 
connected graph with perfect state transfer at time $t_0$; that is, $p(t_0)=1$. Suppose there is 
a small perturbation and the readout time is instead $t_0+h$, where, denoting the smallest and 
largest eigenvalues of $H$ by $\lambda_{1}, \lambda_{n},$ respectively,  $h$ satisfies 
$|h| < \frac{\pi}{\lambda_{n}-\lambda_{1}}$. Then 
the fidelity at the perturbed time $t_0+h$  satisfies the following lower bound:
\[
 p(t_0+h)\geq  \frac14|e^{ih\lambda_1}+e^{ih\lambda_n}|^2.
\]
\end{theorem}

 \begin{proof}
 Let $D=t_0\Lambda=\diag(t_0\lambda_1, \dots, t_0\lambda_n)$ and  
 $\hat{D}=\diag((t_0+h)\lambda_1, \dots, (t_0+h)\lambda_n)$. 
 We find  
\bes
\bra{q_1}e^{i\hat{D}} \ket{q_2} &=& \bra{q_1} (e^{i\hat{D}}e^{-iD})  e^{iD} \ket{q_2}\\
&=& \bra{q_1} (e^{i\hat{D}}e^{-iD}) e^{i\theta} \ket{q_1} = \bra{q_1}M \ket{q_1} \in W(M)
\ees
with $M = \diag(e^{ih\lambda_1}, \dots, e^{ih\lambda_n}) e^{i\theta}$.  
Now, $W(M)$ is the convex hull of $\{e^{i\theta}e^{ih\lambda_1}, \dots, 
e^{i\theta}e^{ih\lambda_n}\}$.  Since $|h\lambda_n - h\lambda_1| < \pi,$ 
there exists an $s \in [0, 2\pi)$ such that
$e^{is}M$ has eigenvalues $e^{i\xi_1}, \dots, e^{i\xi_n}$ with 
$-\pi/2 < \xi_1 \le \cdots \le \xi_n <\pi/2$ with $\xi_1 = -\xi_n$.
Let $e^{is}M = M_1 + iM_2$ such that $M_1 = M_1^\dagger$ and $M_2 = M_2^\dagger$. 
Then $M_1$ has eigenvalues $\cos \xi_1, \dots, \cos \xi_n$ so that
$0 \le \cos \xi_1 = \cos \xi_n \le \cos \xi_j$ for all $j = 2, \dots, n-1$.
As a result, for every unit vector $\ket{q} \in \IC^n$, we have
\bes
|\bra{q} M \ket{q}| &=& |\bra{q}(M_1+iM_2)\ket{q}| \ge |\bra{q} M_1 \ket{q}|
\ge \cos \xi_1 \\
&=& |e^{i\xi_1} + e^{i\xi_n}|/2 = |e^{ih\lambda_1} + e^{ih\lambda_n}|/2.
\ees
Thus, every  point in $W(M)$ has a distance larger than 
$|e^{ih\lambda_1} + e^{ih\lambda_n}|/2$ from 0. 
Consequently, 
\begin{eqnarray*}
p(t_0)-p(t_0+h)&=&
|\bra{q_1} e^{iD} \ket{q_2}|^2 - |\bra{q_1} e^{i\hat D} \ket{q_2}|^2 
\\
&\leq & 1 -  \frac14|e^{ih\lambda_1} + e^{ih\lambda_n}|^2,
\end{eqnarray*}
and the result follows.
 \end{proof}

In fact, in the above proof, one can get a better estimate of  $|\bra{q_1} M \ket{q_1}|$ 
using the  information of $\ket{q_1} = (q_{11}, \dots, q_{n1})^T$
and $M = \diag(e^{ih\lambda_1}, \dots, e^{ih\lambda_n}) e^{i\theta}$; 
namely, for any $s \in \IR$, 
\begin{eqnarray*}
|\bra{q_1}M\ket{q_1}|
&=& \left|\sum_{j=1}^n q_{j1}^2 e^{ih(\lambda_j-s)}\right|
\ge \left|\sum_{j=1}^n q_{j1}^2\cos (h(\lambda_j-s))\right|\\
&\ge& \sum_{j=1}^n q_{j1}^2  - 
\frac{h^2}{2}  \sum_{j=1}^n q_{j_1}^2(\lambda_j-s)^2 \\
&=&  1 - 
\frac{h^2}{2} \sum_{j=1}^n q_{j_1}^2(\lambda_j-s)^2,  \\
\end{eqnarray*}
the second inequality following from the fact that $\cos(x) \ge 1-\frac{x^2}{2}$ for any $x \in \R .$ 

In particular, if we let $s = \lambda_1$ in the above, we obtain a result that is parallel to the bound in \cite{Kay06},  without that paper's extra hypotheses on the Hamiltonian. 

For general $s$, the above implies 
\begin{eqnarray*}
p(t_0)-p(t_0+h)&=&
1-|\bra{q_1}M\ket{q_1}|^2 \\
&=& (1+|\bra{q_1}M\ket{q_1}|)(1-|\bra{q_1}M\ket{q_1}|) \\
&\le& 2 \left(\frac{h^2}{2} \sum_{j=1}^n q_{j_1}^2(\lambda_j-s)^2 \right)\\
&=&h^2\sum_{j=1}^n q_{j_1}^2(\lambda_j-s)^2. \end{eqnarray*}
We summarize these derivations in the following theorem.

\begin{theorem}\label{thm:OurKay}
Let $H$ be either the adjacency matrix or the Laplacian associated to an undirected weighted 
connected graph with perfect state transfer at time $t_0$; that is, $p(t_0)=1$. Suppose there is 
a small perturbation and the readout time is instead $t_0+h$, where, for $\lambda_{1}\leq \cdots\leq  \lambda_{n}$,  $h$ satisfies 
$|h| < \frac{\pi}{\lambda_{n}-\lambda_{1}}$. Then, for any $s\in \R$, 
the transition probability for the perturbed time has the following lower bound:
\[
 p(t_0+h)\geq 1- h^2 \sum_{j=1}^n q_{j_1}^2(\lambda_j-s)^2.
\]
\end{theorem}

  Theorem \ref{thm:OurKay} is an improved bound compared to Theorem \ref{thm:time}. Yet direct use of  Theorem \ref{thm:OurKay} requires one to find all eigenvalues $\lambda_1, \dots, \lambda_n$ of the Hamiltonian $H$, while Theorem \ref{thm:time} requires only that the smallest and the largest eigenvalues are known. For large spin systems, Theorem \ref{thm:time} would then be more practical. However, the following consequence of Theorem \ref{thm:OurKay} yields a lower bound on the fidelity that involves the physical parameters of the Hamiltonian itself, and does not require knowledge of any of the eigenvalues of $H$. 
 
\begin{corollary}\label{cor:s-free} Under the hypotheses of Theorem \ref{thm:OurKay}, we have 
\begin{equation} \label{s-free} 
p(t_0+h) \ge 1 - h^2(\bra{1} H^2 \ket{1} - (\bra{1} H \ket{1})^2).  
\end{equation} 
\end{corollary} 
\begin{proof}
Observe that the quantity $ \sum_{j=1}^n q_{j_1}^2(\lambda_j-s)^2$ is minimised when $s\sum_{j=1}^n q_{j_1}^2 = \sum_{j=1}^n q_{j_1}^2\lambda_j$, i.e.\ when $s =  \sum_{j=1}^n q_{j_1}^2\lambda_j = \bra{1}H\ket{1}.$
The corresponding minimum value is then given by $ \sum_{j=1}^n q_{j_1}^2(\lambda_j-  \bra{1}H\ket{1} )^2 =  \sum_{j=1}^n q_{j_1}^2\lambda_j^2 - 2  \bra{1}H\ket{1}  \sum_{j=1}^n q_{j_1}^2\lambda_j + ( \bra{1}H\ket{1} )^2  \sum_{j=1}^n q_{j_1}^2 = 
(\bra{1} H^2 \ket{1} - (\bra{1} H \ket{1})^2)).$ Inequality (\ref{s-free}) now follows readily from Theorem \ref{thm:OurKay}. 
\end{proof} 

Inequality \ref{s-free} of  Corollary \ref{cor:s-free}  is   fairly accurate in the following sense: In \cite{Steve2015}, the author considers the derivatives of $p$ at time $t_0$ under the hypotheses of PST at $t_0$. In  \cite[Theorem 2.2]{Steve2015}, it is shown that all odd order derivatives of $p$ at $t_0$ are zero, while the second derivative is equal to $-2(\bra{1} H^2 \ket{1} - (\bra{1} H \ket{1})^2))$. From   \cite[Theorem 2.4]{Steve2015},  it follows that the fourth derivative of $p$ at $t_0$ is positive.  It now follows that for all $h$ with $|h|$ sufficiently small, there is a $c>0$ such that $p(t_0+h) = 1 - h^2(\bra{1} H^2 \ket{1} - (\bra{1} H \ket{1})^2) + ch^4 + O(h^6) $.  In other words, for small $h$, the lower bound of  Corollary \ref{cor:s-free}   is accurate to terms in $h^3$.

We now consider the case where equality holds in the bound of 
Theorem \ref{thm:time} when $H$ is the adjacency matrix of a connected weighted graph. For concreteness, 
suppose that $H$ is of order $n$ and that there is perfect state transfer at time $t_0$. Suppose 
further that for some $h$ with $|h| < \frac{\pi}{\lambda_{n}-\lambda_{1}}$ we have 
$p(t_0)-p(t_0+h)= 1-\frac14|e^{ih\lambda_1}+e^{ih\lambda_n}|^2.$ Denote the multiplicity of 
$\lambda_1$ by $k$, and recall that $\lambda_n,$ as the Perron value of $H$ (that is, the unique maximal eigenvalue as per the Perron–Frobenius theorem), is necessarily simple. 
Examining the proof of Theorem  \ref{thm:time}, it follows that the $\ket{q_1}$ can only have 
nonzero entries in positions corresponding to the eigenvalues $\lambda_1$ and $\lambda_n$, and that 
further the entry in the position corresponding to $\lambda_n$, the $j$--th position say, 
must be $\pm \frac{1}{\sqrt{2}}$. Since the entries of $\ket{q_2}$ can only differ from the 
corresponding entries of $\ket{q_1}$ by a sign, we deduce that the 
$j$--th entry of $\ket{q_2}$ must also be $\pm \frac{1}{\sqrt{2}}$. Observe that since every column 
of  $Q^T$ has $2$--norm equal to $1$, the $j$--th column of $Q^T$ has nonzero entries only in its 
first two rows.  
But the $j$--th column of $Q^T$ is a Perron vector for $H$---that is, an eigenvector corresponding to the (positive and dominant) Perron eigenvalue of $H$---and so it cannot have any zero entries (again by the  Perron–Frobenius theorem). 
We thus deduce that $n$ must be $2$, and that $H$ must be a positive scalar multiple of $\bpm 
0&1\\1&0\epm$. 

Conversely, suppose that $H$ is a positive scalar multiple of $\bpm 0&1\\1&0\epm$, and without loss 
of generality we assume that  $H=\bpm 0&1\\1&0\epm .$ Then 
\bes
e^{itH}&=&\sum_{j=0}^\infty\frac{i^j}{j!}H^jt^j=\sum_{j \textnormal{ even}}\frac{(i)^j}{j!}t^jI
+\sum_{j \textnormal{ odd}}\frac{(i)^j}{j!}t^jH\\
&=&\bpm \cos t&i\sin t\\ i\sin t&\cos t\epm.
\ees
At time $\frac{\pi}2$ we have $\bpm 0&i\\i&0\epm$ and so we have PST (since $|i|=1$). At time 
$\frac{\pi}2+h$, the $(1,2)$ entry is $i\cos h=\frac12(e^{ih}+e^{-ih})$, 
and so the  bound in Theorem \ref{thm:time} is attained. 

Although Theorem \ref{thm:time} is true for either adjacency matrices or Laplacians, we can adapt 
the technique of Theorem \ref{thm:time} slightly to produce an improved bound on the fidelity in the 
setting of the Laplacian matrix, since we have more information at hand. 

\begin{theorem}\label{lapl:time} Let $L$ be the Laplacian matrix of a connected weighted graph on $n 
\ge 3$ vertices. Denote the eigenvalues of $L$ by $0 \equiv \lambda_1 < \lambda_2 \le \ldots \le 
\lambda_n.$ 
Suppose that there is perfect state transfer at time $t_0$; that is, $p(t_0)=1$. Suppose there is a 
small timing perturbation, and the readout time is instead $t_0+h$, where  $h$ satisfies 
$|h|< \frac{\pi}{\lambda_n-\lambda_2}$. Then 
\begin{eqnarray} \label{lapl_bd} \nonumber 
p(t_0+h) &\ge& 1 - \frac{(n-1)^2(1-\cos ((\lambda_n-\lambda_2)h))}{2n^2}   \nonumber \\
&& - \frac{(n-1)(2- \cos (\lambda_2 h)  - \cos (\lambda_n h))}{n^2}  \nonumber \\
&& - \frac{(\cos (\lambda_2 h)  - \cos (\lambda_n h) )^2}{2n^2(1- \cos((\lambda_n-\lambda_2)h)}.
\end{eqnarray}
\end{theorem} 
\begin{proof} 
We note that  the normalised   all--ones vector 
$\frac{1}{\sqrt{n}} \sum_{i=1}^{n}\ket{i}$ is a null vector for $L$.
Then $L = Q^T\diag(\lambda_1, \dots, \lambda_n)Q$ such that 
$\ket{q_1}$ has the form $(1/\sqrt n, x_2, \dots, x_n)^T$ where $\sum_{j=2}^nx_j^2 = \frac{n-1}{n}$. 
Mimicking the proof of Theorem \ref{thm:time}, we find that $p(t_0+h)$ is bounded below by
$$
\min \Big|\frac{1}{n} + \bra{z} \diag(e^{ih\lambda_2}, \dots, e^{ih\lambda_n}) \ket{z} \Big|^2,$$ 
where the minimum is taken over all unnormalised $\ket{z} \in \IR^{n-1}$ such that $\langle z|z\rangle=\frac{n-1}{n}$. 
From elementary geometric considerations (in short, the $e^{ih\lambda_j}$ are points on the unit 
circle and so the minimum will be attained by taking a convex combination of the smallest and 
largest values, namely $e^{ih\lambda_2}$ and $e^{ih\lambda_n}$), we find that in fact 
\bes
\min \Big |\frac{1}{n} + \bra{z}\diag(e^{ih\lambda_2}, \dots, e^{ih\lambda_n}) \ket{z} \Big |^2 = \\ 
\min_{0 \le \alpha \le \frac{n-1}{n}} \Big|\frac{1}{n}+ \alpha e^{ih\lambda_2} + \left(\frac{n-1}{n}-\alpha \right) e^{ih\lambda_n} \Big|^2.
\ees 
A routine calculus exercise (the minimum corresponds to  
\bes
\alpha = \frac{n-1}{2n} +
\frac{\cos(h\lambda_n) - \cos(h\lambda_2)}{2n(1-\cos(h(\lambda_2-\lambda_n)))}\Big)
\ees  shows that this last quantity is given by the right hand side of (\ref{lapl_bd}).
\end{proof} 

\begin{example} Suppose that $n$ is divisible by $4$, and consider the unweighted graph on  
vertices $1, \ldots, n$, say $G$, formed by deleting the edge between vertices $1$ and $2$ 
from the complete graph on $n$ vertices. Let $L$ denote the Laplacian matrix for $G$, and note that $L$ has three eigenvalues: 
$0$, with corresponding eigenprojection matrix $\frac{1}{n}J$ (where $J$ is the all ones matrix), $n-2$ with eigenprojection matrix 
$\frac{1}{2}(\ket{1}-\ket{2})(\bra{1}-\bra{2}),$ and $n$ with eigenprojection matrix 
$$\left[\begin{array}{c|c} \frac{n-2}{2n} J_2 & -\frac{1}{n} J_{2,n-2} \\ 
\hline -\frac{1}{n} J_{n-2,2}&I-\frac{1}{n}J_{n-2} \end{array}\right].$$  
It is shown in \cite{K-Sev} that, using $L$ as the Hamiltonian, there is perfect state transfer 
from vertex $1$ to vertex $2$ at time $\frac{\pi}{2}.$ 

Using the eigenvalues and eigenprojection matrices above, we find that for any $h,$ the fidelity 
at  time $\frac{\pi}{2} + h$ is given by 
$$p \left(\frac{\pi}{2} + h \right) = \Big | \frac{1}{n} + \frac{1}{2} e^{ih(n-2)} 
+ \frac{n-2}{2n}e^{ihn} \Big|^2.$$ 
This last expression can be simplified to yield 
\bes
p\left(\frac{\pi}{2} + h \right) &=& 1 - \frac{n-2}{2n} (1 - \cos(2h)) 
- \frac{1}{n}(1-\cos((n-2)h)) \\
&&- \frac{(n-2)}{n^2}(1-\cos(nh)).
\ees 
An uninteresting computation reveals that $p(\frac{\pi}{2}+h)$ exceeds 
 the lower bound of (\ref{lapl_bd}) in the amount of $$\frac{(\cos((n-2)h) 
 - \cos(nh)+1- \cos(2h))^2}{2n^2(1-\cos(2h))}.$$
We note in passing that  $ \frac{(\cos((n-2)h) - \cos(nh)+1- \cos(2h))^2}{2n^2(1-\cos(2h))}$ is 
asymptotic to $h^2$ as $h \rightarrow 0.$
\end{example}

\section{Sensitivity with respect to edge weights}\label{sec:edgeweight}
 
The motivation for studying edge weight perturbation is similar to that for readout time: although 
lab equipment can be calibrated to high precision, small errors (even machine epsilon) will affect 
the state transition probability. We want to bound this effect so that small perturbations in edge 
weight do not drastically reduce transition probability.

Suppose we have PST between vertices $j$ and $k$ at time $t_0$.  As in section \ref{sec:time}, 
without loss of generality we consider the $(1,2)$ entry of our matrix under consideration. Here, 
we keep the time constant at $t_0$, and perturb edge weights; combining the techniques employed in section 
\ref{sec:time} for readout time with the techniques in this section to 
obtain a bound for the situation when both readout time and edge weights are perturbed is beyond the scope of this paper.
 
Let $H \in M_n$  be a real symmetric matrix representing the Hamiltonian of our system.   
Suppose  $\bra{1}e^{it_0H} \ket{2}$ has modulus 1; this is the case 
for PST with $H$ either an adjacency or Laplacian matrix. 
Consider now $\hat{H}=t_0H+H_0$ where $H_0$ is a matrix representing small perturbations of edge 
weights and we have absorbed the time component into $H_0$ since we are keeping time fixed. 
Mathematically, we would like to find a perturbation bound for 
$|\bra{1}e^{itH} \ket{2}|^2 - |\bra{1}e^{i(tH + H_0)} \ket{2}|^2$ or, when $t=t_0$, 
$$1- |\bra{1}e^{i(t_0H + H_0)} \ket{2}|^2  $$ 
for a symmetric matrix $H$ with sufficiently small $H_0$,
say, measured by the operator norm $\|H_0\|$ or the Frobenius norm $\|H_0\|_F$. 

Note that the entries of  the matrix
$H_0$ represent individual edge weight errors, so our approach allows for individual edge 
weight perturbations rather than simply an overall (global) edge weight perturbation (where all edge weights are perturbed by e.g.\ 0.0001 in the same direction) or a single edge weight perturbation (where all other edge weights remain unperturbed); the latter case was the (rather restrictive) situation considered in \cite{Steve2015}.

We begin with the following.

\begin{theorem} \label{bound-GVL}
Suppose a perfect state transfer occurs at time $t_0$,
and $\hat H = t_0H+ H_0$, with a small nonzero
perturbation $H_0$. Then
$$\|e^{i(Ht_0 + H_0)}-e^{iHt_0}\| \le \|H_0\| e^{\|H_0\|}.$$
Consequently,
\begin{eqnarray}\label{exp-bound}
1 - |\bra{1} e^{i(Ht_0+H_0)} \ket{2}|^2 &\le& 2 \|H_0\|e^{\|H_0\|} - \|H_0\|^2e^{2\|H_0\|}\nonumber\\
&\le& 2\|H_0\| + \|H_0\|^2 - \|H_0\|^3.
\end{eqnarray}
\end{theorem}

\begin{proof} Set $\Delta t_0 = H_0$. Using
the result in \cite[p.532]{GVL} 
and the fact that $H$ is Hermitian, we have 
$$\|e^{i(H+\Delta)t_0}-e^{iHt_0}\| \le t_0 \|\Delta\| e^{t_0\|\Delta\|}=  
\|H_0\| e^{\|H_0\|}.$$
Consequently, 
$$|\bra{1}e^{it_0H}\ket{2}|-|\bra{1}e^{i(t_0H + H_0)}\ket{2}|\le \|H_0\| e^{\|H_0\|}$$
so that 
$$1 - \|H_0\| e^{\|H_0\|} \le |\bra{1}e^{i(t_0H+H_0)}\ket{2}|.$$
Squaring both sides and rearranging terms, we have
\begin{eqnarray*}
&&1- |\bra{1} e^{i(t_0H+H_0)}\ket{2}|^2 \\
&\le& 2\|H_0\| e^{\|H_0\|} - \|H_0\|^2 e^{2\|H_0\|}\\
&=& 2\|H_0\|(1 + \frac{\|H_0\|}{1!} + \frac{\|H_0\|^2}{2!} + \cdots ) \\
&&  - \|H_0\|^2 (1 +  \frac{2\|H_0\|}{1!} + \frac{(2\|H_0\|)^2}{2!} + \cdots ) \\
&\le& 2\|H_0\| + \|H_0\|^2 -  \|H_0\|^3. 
\end{eqnarray*} 
so that (\ref{exp-bound}) holds. 
\end{proof}

We note that the estimate $\|e^{i(Ht_0 + H_0)}-e^{iHt_0}\| \le \|H_0\| e^{\|H_0\|}$ of Theorem \ref{bound-GVL} can be reasonably accurate. For example, suppose $H$ is the adjacency matrix of a connected weighted graph yielding perfect state transfer at time $t_0$. Let $\ket{v}$ denote the positive Perron vector of $H$ with norm one, and suppose that $H_0$ has the  form $\epsilon\ket{v}\bra{v}$ for some small $\epsilon >0.$ Then  $\|e^{i(Ht_0 + H_0)}-e^{iHt_0}\|  = |e^{i \epsilon}-1|$ while $ \|H_0\| e^{\|H_0\|} = \epsilon e^{\epsilon}$, so that $$\frac{\|e^{i(Ht_0 + H_0)}-e^{iHt_0}\|}{ \|H_0\| e^{\|H_0\|}} \rightarrow 1$$ as $\epsilon \rightarrow 0^+.$ 

\medskip
If we have additional information about the matrix $t_0H$, 
we may be able to produce some better bounds as shown in Theorem \ref{thm:ewps-1} below. Before presenting the theorem, we require a preliminary proposition which is intuitively clear. Its proof consists of elementary linear algebra manipulation techniques; we give the proof for completeness. 

\begin{proposition}\label{prop:decomp} 
Suppose there is perfect state transfer at time $t_0$; that is, $|\bra{1}e^{it_0H}\ket{2}| = 1$.
Then for some $\theta \in \IR$, $t_0H =  \tilde{Q}^T\tilde{D}\tilde{Q}+\theta I$, where 
$$\tilde{D} = \pi \diag(r_1, \dots, r_\ell, r_{\ell+1}, \dots, r_m, r_{m+1}, \dots, r_n)$$
such that $r_1\geq  \cdots\geq  r_\ell$ are positive even integers, and 
$r_{\ell+1}\geq  \cdots\geq r_m$ are positive odd integers.
Further, the first two rows of $\tilde{Q}^T$ can be taken to have the form
$(x_1, \dots, x_m, 0, \dots, 0)$ and $(x_1,\dots, x_\ell, -x_{\ell+1}, \dots, -x_m, 0, \dots, 0)$
satisfying $x_1, \dots, x_m \ge 0$.
\end{proposition}

\begin{proof} As in section \ref{sec:time}, let $H = Q^T\Lambda Q$ for a real diagonal matrix $\Lambda$
and $Q$ an orthogonal matrix. Since we are focusing here on edge weights rather than time, let us   consider $D=t_0\Lambda=\diag(t_0\lambda_1, \dots, t_0\lambda_n)$. 
Suppose the first two columns of $Q$ are $\ket{q_1}$ and $\ket{q_2}$.
Then $|\bra{q_1} e^{iD} \ket{q_2}| = 1$ implies that
$e^{iD}  \ket{q_2}= e^{i\theta} \ket{q_1}$ or, equivalently, 
$e^{-i\theta} e^{iD} \ket{q_2} = \ket{q_1}$ for some $\theta\in \R$. Any zero entries of $\ket{q_1}$ show up correspondingly as zero entries of $\ket{q_2}$. For a suitable permutation matrix $P_1$ we can replace $(D,Q)$ by $(P_1DP_1^T, P_1Q)$ so that the zero entries of $P_1\ket{q_1}$ all occur in the last $n-m$ entries, for some $0<m\leq n$. 
In this way, we may assume that $P_1(D - \theta I)P_1^T$ is a diagonal matrix
with diagonal entries of the form $s_1 \pi, \dots, s_m\pi, *, \dots, *$ for some
integers $s_1, \dots, s_m$. The asterisks in the $(m+1, m+1)$ up to $(n,n)$ entries of the diagonal matrix  $P_1(D - \theta I)P_1^T$ represent unknown constants, corresponding to the zero entries (if any) of $P_1\ket{q_1}$.  We can replace $\theta$ by $\theta - 2 s \pi$ for a sufficiently large
integer $s$ so that we may assume that $s_1, \dots, s_m$ are positive integers.

Next, for a suitable permutation matrix $P_2$ we can replace the pair $(P_1DP_1^T, P_1Q)$ by $(P_2P_1DP_1^TP_2^T, P_2P_1Q),$  so that 
$P_2P_1(D-\theta I)P_1^TP_2^T = \pi \diag(r_1, \dots, r_n)$ with $r_1\geq  \cdots\geq r_\ell$  even, 
$r_{\ell+1}\geq  \cdots\geq r_m$  odd, and $r_{m+1}, \dots, r_n$ unknown constants; note that we still have
$t_0H = (P_2P_1Q)^T (P_2P_1DP_1^TP_2^T) (P_2P_1Q)$.
Further, we may replace the pair  $(P_2P_1DP_1^TP_2^T, P_2P_1Q),$ by  $(SP_2P_1DP_1^TP_2^TS, SP_2P_1Q),$ for some diagonal orthogonal matrix
(i.e.\ a signature matrix) $S$ such that the first column of $SP_2P_1Q$, namely 
$SP_2P_1\ket{q_1} = (x_1, \dots, x_m, 0, \dots, 0)^T$,  satisfies $x_1, \dots, x_m\geq 0$.
Now, $SP_2P_1e^{-i\theta}e^{iD} \ket{q_2} = SP_2P_1\ket{q_1}$ implies that 
$$SP_2P_1\ket{q_2}= (x_1, \dots, x_\ell, -x_{\ell+1}, \dots, -x_m, 0, \dots, 0)^T.$$ Relabelling $\tilde{D}=SP_2P_1(D-\theta I)P_1^TP_2^TS$ and $\tilde{Q}=SP_2P_1Q$ for simplicity, the result now follows.
\end{proof}

\begin{theorem}\label{thm:ewps-1} 
Suppose a perfect state transfer occurs at time $t_0$,
and $\hat H = t_0H+ H_0$, with a small nonzero
perturbation $H_0$. Furthermore, assume that value $m$ in Proposition {\rm \ref{prop:decomp}}
equals $n$. Then 
\begin{equation*}
1- |\bra{1}e^{i \hat H}\ket{2}|^2 \le \frac{2\|H_0\|_F^2}{(\pi - \|H_0\|)^2} + \|H_0\|^2 + O(\|H_0\|^3).
\end{equation*}
  \end{theorem}

\begin{proof}
Let $t_0H = Q^TDQ+\theta I$, where $D =\tilde{D} = \pi \diag(r_1, \dots, r_\ell, r_{\ell+1}, \dots, r_m, r_{m+1}, \dots, r_n)$  as in Proposition 
\ref{prop:decomp}. We drop the tilde for notational simplicity. Relabel $d_j=\pi r_j$ for $j=1, \dots, n$, so that $D=\diag(d_1, \dots, d_n)$ is such that
the first $\ell$ entries are even multiples of $\pi$ and the last $n-\ell$ entries are  
odd multiples of $\pi$. Let $J = e^{iD} = I_\ell \oplus -I_{n-\ell}$. 

Suppose $\hat H = t_0H+ H_0 = \hat Q^T \hat D \hat Q+\theta I$.
By a suitable choice of $\hat Q$, we may assume that there is 
a permutation matrix $P$ such that 
both $P^TDP$ and $P^T\hat DP$ have diagonal entries arranged in 
descending order. Then (e.g., see \cite[p.101,(IV.62)]{B})
\bes
\|D-\hat D\| &=& \|P^TDP-P^T \hat D P\| \\
&\le& \|Q^TDQ - \hat Q^T \hat D \hat Q\| = \|H_0\|
\ees
and hence
\begin{equation}\label{D-hatD}
\|e^{iD} - e^{i\hat D}\| \le \|D-\hat D\|
\le \|H_0\|. 
\end{equation}


Let  $V$ be an orthogonal matrix close to $I$, and consider the power series $\log(V) = - \sum_{j=1}^{\infty} \frac{1}{j}(I-V)^j.$ Setting $K=\log(V),$ we have $e^K=V.$ It follows that $I=VV^T = e^K e^{K^T} = e^{K+K^T},$ where the last equality comes from the fact that $K$ commutes with $K^T$. We deduce that $K^T=-K,$ i.e. $K$ is skew--symmetric. We will use this idea in what follows.

If $H_0$ is small, we may assume that the differences between the eigenspaces of $\hat H$ and $t_0H$
are small so that  $\hat D$ is close to $D$, and $\hat Q Q^T$ is close to $I$ 
(see \cite[Section VII.3]{B}). As a result, we can write
$\hat D - D = w D_1$ and $e^{wK} = \hat Q Q^T$ for 
a small positive number $w$, a diagonal matrix $D_1$ and a skew-symmetric matrix $K$
such that $\max\{\|D_1\|, \|K\|\}=1$ (the norm condition is required so that the terms like $K^3$, $D_1^3$ can be lumped into the $O(w^3)$ term below). 
Note that it is possible that $D_1 = 0$ or $K = 0$ but not both as $H_0 \ne 0$. We emphasize that $\log(\hat QQ^T)$ is skew--symmetric, from the above remark about a matrix $V$.

Now $\hat H= Q_w^T(D+wD_1)Q_w$ where we write $Q_w = e^{wK}Q$, using the subscript $w$ here to emphasize the dependence on some small positive number $w$. 
Using the power series expansion of $e^{i\hat H}$ and the fact that $K = -K^T$, we get
\begin{equation*} 
\begin{aligned}
&  e^{-i\theta}\bra{1}e^{i\hat H}\ket{2}=\bra{1}Q_w^Te^{i\hat D}Q_w\ket{2}\\
&=\bra{1}[Q^T(I+wK+\frac{1}{2}w^2K^2)^Te^{iD}
(I+iwD_1\\
&\quad-w^2\frac{1}{2}D_1^2)(I+wK+\frac{1}{2}w^2K^2)]Q\ket{2}+O(w^3)\\
&=\bra{1}Q^Te^{iD}Q\ket{2}+
w\bra{1}Q^T[K^Te^{iD}+e^{iD}iD_1\\
&\quad +e^{iD}K]Q\ket{2} +\frac{1}{2}w^2 \bra{1}Q^T [ (K^T)^2e^{iD} \\
&\quad+
e^{iD}K^2 - e^{iD}D_1^2]  Q\ket{2}\\
&\quad + w^2\bra{1}Q^T[K^Te^{iD}iD_1+K^Te^{iD}K \\
&\quad+e^{iD}iD_1K]Q\ket{2}+O(w^3).
\end{aligned}
\end{equation*}
By the facts that $J\ket{q_2} = \ket{q_1}$ and $\bra{q}K\ket{q}=0$ for any vector 
$\bq$,  the above expression becomes
\begin{equation*} 
\begin{aligned}
&e^{-i\theta}\bra{1}Q_w^Te^{i\hat D}Q_w\ket{2}=
1+ w(\bra{q_1}K^T\ket{q_1}+\bra{q_2}K\ket{q_2})\\
& \quad +\frac{1}{2}w^2 \bra{q_1}((K^T)^2-D_1^2+ JK^2J+2K^TJKJ)\ket{q_1}\\
&\quad +iw\bra{q_1}D_1\ket{q_1}+ iw^2\bra{q_1}(K^TD_1+JD_1KJ)\ket{q_1}+O(w^3)\\
&= 1+ \frac{1}{2}w^2\bra{q_1}((K^T)^2 - D_1^2+ JK^2J+2K^TJKJ)\ket{q_1}\\
&\quad+iw\bra{q_1}D_1\ket{q_1}+iw^2\bra{q_1}(K^TD_1+D_1JKJ)\ket{q_1} +O(w^3).\\
\end{aligned}
\end{equation*}
Let $x(w)=\Re(\bra{1}Q_w^Te^{i\hat D}Q_w\ket{2}),$ and 
$y(w)=\Im(\bra{1}Q_w^Te^{i\hat D}Q_w\ket{2})$.
Then 
\begin{equation*} 
\begin{aligned}
&|\bra{1}e^{i(t_0H+H_0)}\ket{2}|^2 - |\bra{1}e^{it_0H}\ket{2}|^2  
=|\bra{1}e^{i(t_0H+H_0)}\ket{2}|^2-1\\
&=x(w)^2 + y(w)^2 -1\\
&  = 
w^2[\bra{q_1}(K^2+ JK^2J - 2KJKJ - D_1^2)\ket{q_1} \\
&\quad + (\bra{q_1}D_1\ket{q_1})^2]
+ O(w)^3 \\
&   = 
-w^2 \big( \|(KJ - JK)\ket{q_1}\|^2 + \|D_1\ket{q_1}\|^2 \\
&\quad - (\bra{q_1}D_1\ket{q_1})^2 \big) 
+ O(w^3).
\end{aligned}
\end{equation*}
For the last equality in the above expression, we use the fact that, although $JKJK\neq KJKJ$, 
it is true that  $\bra{q_1}JKJK\ket{q_1}=\bra{q_1}KJKJ\ket{q_1}$, which is all that is required  here. 

Thus, 
if $wK = \begin{bmatrix}K_{11} & K_{12} \cr -K_{12}^T & K_{22}\cr\end{bmatrix}$, then 
$wKJ - wJK  =  \begin{bmatrix}O & -2K_{12} \cr -2 K_{12}^T & O \cr\end{bmatrix}$
and hence 
\begin{equation}\label{ineq1}
\|(wKJ - wJK)\ket{q_1}\|^2 \le 4 \|K_{12}\|^2.
\end{equation}
Now, 
$$\hat Q Q^T = e^{wK} = I + wK + \frac{(wK)^2}{2!} + 
\frac{(wK)^3}{3!} + \cdots.$$ 
So, $w K \approx (\hat Q Q^T - Q \hat Q^T)/2$. 

As a result, if $\hat Q Q^T \equiv V 
=  \begin{bmatrix}V_{11} & V_{12} \cr V_{21}& V_{22} \cr\end{bmatrix}$,
then  $V_{12}$ and $V_{21}$ have the same nonzero singular values (there are a number of ways 
of seeing this; perhaps the simplest is to note that since $\hat Q Q^T$ is orthogonal, it follows 
that $V_{11}V_{11}^T+V_{12}V_{12}^T= I$ and $V_{11}^TV_{11}+V_{21}^TV_{21}= I$; that is, 
$V_{11}V_{11}^T=I-V_{12}V_{12}^T$ and $V_{11}^TV_{11}=I-V_{21}^TV_{21}$, from which we see 
$V_{12}V_{12}^T$ and $V_{21}^TV_{21}$ have the same eigenvalues).
Note that
\begin{eqnarray*}
\|H_0\|_F^2 &=&
\|Q^T D Q - \hat Q^T \hat D \hat Q\|_F^2 =  \|V D - \hat D V\|_F^2 \\ 
&=& \sum_{j,k} (d_k   - \hat d_j)^2 v_{jk}^2.
\end{eqnarray*}
The reverse triangle inequality gives us
\begin{eqnarray*}
|d_k -\hat d_j|&=&|d_k -d_j+d_j-\hat d_j|\\
&\geq& |d_k -d_j|-|d_j-\hat d_j|\\
&\geq &\pi-\|H_0\|,
\end{eqnarray*}
where the last inequality follows assuming either $1\le j \le \ell < k \le n$ or $1\le k \le \ell < j \le n$. 
We now have
\begin{eqnarray*}
\|H_0\|_F^2 &=&\sum_{j,k} (d_k -\hat d_j)^2 v_{jk}^2\\
&\ge&   \sum_{1\le j \le \ell < k \le n} (d_k  - \hat d_j)^2 v_{jk}^2\\
&&\quad + \sum_{1\le k \le \ell < j \le n} (d_k   - \hat d_j)^2 v_{jk}^2\\ 
&\ge& (\pi - \|H_0\|)^2 (\|V_{12}\|_F^2 + \|V_{21}\|_F^2).\end{eqnarray*}
It follows that 
$$\|V_{12}\|_F^2 + \|V_{21}\|_F^2 
\le \frac{\|H_0\|_F^2}{(\pi - \|H_0\|)^2}.$$
Hence,
\begin{eqnarray}\label{ineq2}
\|K_{12}\| &\le& (\|V_{12}\| + \|V_{21}\|)/2 = \|V_{12}\| \le \|V_{12}\|_F \nonumber\\
&\le& \frac{\|H_0\|_F}{\sqrt 2(\pi - \|H_0\|)}.
\end{eqnarray} 
As a result, from (\ref{ineq1}) and (\ref{ineq2}) we have 
$$w^2 \|(KJ - JK)\ket{q_1}\|^2 
\le 4 \|K_{12}\|^2 \le
\frac{2\|H_0\|_F^2}{(\pi - \|H_0\|)^2}.$$

A result of Mirsky \cite{Mirsky} states that for any Hermitian matrix $M$, the eigenvalue spread for $M$ is equal to $2 \max |\bra{u}A\ket{v}|,$ where the maximum is taken over all pairs of orthonormal vectors $\ket{u}$ and $\ket{v}$. 
Consequently, for any symmetric matrix $A$, if $\{\ket{u},\ket{v}\}$ is an orthonormal set, 
then 
$$2|\bra{u}A\ket{v}|  \le \lambda_{n}(A) - \lambda_{1}(A),$$ 
where we recall  $\lambda_1$ is the minimum eigenvalue and $\lambda_n$ is the maximum eigenvalue.
In particular, if we set $A = wD_1$, $\ket{u} = \ket{q_1}$, and $A \ket{q_1} = \mu_1 \ket{q_1} + \mu_2 \bq$
such that $\{\ket{q_1}, \ket{q}\}$ is an orthonormal set, 
then 
\bes 
&&\|wD_1\ket{q_1}\|^2 - (\bra{q_1}wD_1\ket{q_1})^2 \\
&=& \mu_2^2 
= |\bra{q} (wD_1) \ket{q_1}|^2\\
&\le& ((\lambda_{n}(wD_1) - \lambda_{1}(wD_1))/2)^2.
\ees
By (\ref{D-hatD}), and recalling that $D_1$ is diagonal, we have 
$$((\lambda_{n}(wD_1) - \lambda_{1}(wD_1))/2)^2
\le  \|\hat D-D\|^2 \le \|H_0\|^2.$$ 
\end{proof}

\vskip .3in

Consider the bounds of Theorems \ref{bound-GVL} and \ref{thm:ewps-1} when the perturbing matrix $H_0$ is small. The upper bound in the former result is $2||H_0|| + ||H_0||^2 - ||H_0||^3$, while the upper bound in the latter result is 
$\frac{2\|H_0\|_F^2}{(\pi - \|H_0\|)^2} + \|H_0\|^2 + O(\|H_0\|^3).$ Thus we find that, neglecting terms of order $\|H_0\|^3$, the bound of Theorem \ref{thm:ewps-1}  is sharper than that of  Theorem \ref{bound-GVL} 
provided that 
\begin{equation}\label{3.1-3.4}
\frac{\|H_0\|_F^2}{(\pi - \|H_0\|)^2} < \|H_0\|. 
\end{equation} 
Suppose for concreteness that $H_0$ has rank $r$. Recalling that $\|H\|_F^2 \le r \|H_0\|^2,$ we find that in order for (\ref{3.1-3.4}) to hold, it is sufficient that 
$r\|H_0\| < (\pi-\|H_0\|)^2,$ or equivalently, that $\|H_0\| < \frac{2\pi + r - \sqrt{4 \pi r + r^2}}{2}.$ It now follows that for all sufficiently small $H_0$, the bound of Theorem \ref{thm:ewps-1} is improvement upon that of Theorem \ref{bound-GVL}. 
Thus, in the case that the more restrictive hypothesis of Theorem \ref{thm:ewps-1} holds, we get a better estimate from that result than from Theorem \ref{bound-GVL}.   

\begin{example} Here we give a small numerical example illustrating the main result of Theorem \ref{thm:ewps-1}. Consider the $10 \times 10$ symmetric tridiagonal matrix $H$ with $h_{j,j+1}=h_{j+1,j}=\sqrt{j(10-j)}, j=1,\ldots, 9$ and all other entries equal to $0$. It is known that for this $H$, there is perfect state transfer from $1$ to $10$ at time $t_0=\frac{\pi}{2},$ 
with the $(1,10)$ entry of $e^{i t_0 H}$ equal to $i$. 

Next, we consider the perturbing matrix $$H_0= 10^{-5} \times \left[\begin{array}{cccccccccc}0&1&0&0&0&0&0&0&0&0\\1&0&0&0&0&0&0&0&0&0\\0&0&0&0&0&0&0&0&0&0\\0&0&0&0&0&0&0&0&0&0\\
0&0&0&0&0&0.533&0&0&0&0\\0&0&0&0&.533&0&0&0&0&0\\0&0&0&0&0&0&0&0&0&0\\0&0&0&0&0&0&0&0&0&0\\0&0&0&0&0&0&0&0&0&1\\0&0&0&0&0&0&0&0&1&0\end{array}\right].$$
Setting $\hat{H}=t_0H+H_0,$ a couple of MATLAB$^\copyright$ computations yield 
$1- |\bra{1}e^{i\hat{H}}\ket{10}|^2 \approx 0.02497 \times 10^{-9}$ and 
$ \frac{2\|H_0\|_F^2}{(\pi - \|H_0\|)^2}  + \|H_0\|^2 \approx 0.19257 \times 10^{-9}$. We note that the ratio of the latter to the former is approximately $7.7110$. 

\end{example}

\section{Conclusion}
We have obtained bounds on the probability of state transfer for a perturbed system, where either readout time or edge weights have been perturbed. By considering such timing and manufacturing errors, our results are physically relevant and more consistent with reality. We worked in the most general setting where the adjacency matrix $A$ (or, alternatively, the Laplacian $L$) was arbitrary, and the perturbations themselves were arbitrary. More precise bounds can be obtained by considering more structured perturbations. 
Furthermore, it would be of interest to combine readout time error  with edge weight error to create one bound encompassing both types of perturbations. Finally, we note that our analysis assumed perfect state transfer (PST). While there are a number of classes of graphs exhibiting PST, it is of interest to allow for pretty good state transfer (PGST) and perform a similar analysis with respect to readout time and edge weight errors; note that the numerical evidence reported in Examples 3.16 and 3.17 in  \cite{Steve2015} suggests that the fidelity may not be so well--behaved under perturbation of edge weights in the  PGST setting. 
  Analysis in the PGST case would require alternate techniques, however, since our arguments hinged on the modulus of the $(1,2)$ entry of our matrix $e^{itH}$ being exactly 1, which facilitates the key observation that $e^{i\theta} \ket{q_1} =  e^{it_0\Lambda} \ket{q_2}$. We leave these as open problems for further study.

\begin{acknowledgments} 
 W.G.\ was supported through a NSERC Undergraduate Student Research Award. S.K.\ and S.P.\ are supported by NSERC Discovery Grants. X.Z.\ is supported by the University of Manitoba's Faculty of Science and Faculty of Graduate Studies. C.-K.L.\ is supported 
 by USA NSF grant DMS 1331021, Simons Foundation Grant 351047, and NNSF of China Grant 11571220. The authors wish to thank the anonymous referee for useful comments and for pointing out additional relevant literature; in particular, for suggesting the result Corollary \ref{cor:s-free}. 
\end{acknowledgments}

\end{document}